\documentclass[numbers=enddot,12pt,letterpaper,final,onecolumn,notitlepage]{scrartcl}
\usepackage[headsepline,footsepline,manualmark]{scrlayer-scrpage}
\usepackage[height=8.5in,letterpaper,hmargin={1.1in,1.1in}]{geometry}
\usepackage[all,cmtip]{xy}
\usepackage{amsfonts}
\usepackage{amssymb}
\usepackage{amsmath}
\usepackage{comment}
\usepackage{needspace}
\usepackage{color}
\usepackage[breaklinks=true]{hyperref}
\usepackage[sc]{mathpazo}
\usepackage[T1]{fontenc}
\usepackage{amsthm}

\makeatletter
\DeclareOldFontCommand{\rm}{\normalfont\rmfamily}{\mathrm}
\DeclareOldFontCommand{\sf}{\normalfont\sffamily}{\mathsf}
\DeclareOldFontCommand{\tt}{\normalfont\ttfamily}{\mathtt}
\DeclareOldFontCommand{\bf}{\normalfont\bfseries}{\mathbf}
\DeclareOldFontCommand{\it}{\normalfont\itshape}{\mathit}
\DeclareOldFontCommand{\sl}{\normalfont\slshape}{\@nomath\sl}
\DeclareOldFontCommand{\sc}{\normalfont\scshape}{\@nomath\sc}
\makeatother

\newif\iflongmode  \longmodefalse

\def\fnote{\footnote}

\def\Abstract#1\par{\vskip 50pt\begin{abstract}#1\end{abstract}}

\def\Section#1#2.{\section{#2.}\label{#1se*}\iflongmode\setcounter{equation}{0}\fi}

\def\Subsection#1#2.{\subsection{#2.}\label{#1se*}}

\newtheorem{Theorem}{Theorem}

\newtheorem{Lemma}[Theorem]{Lemma}

\def\Th#1{\stdeclare{Theorem}{#1}} \def\th{\stwrite{Theorem}}
\def\ths#1#2{Theorems \stn{#1} and~\stn{#2}}
\def\Lm#1{\stdeclare{Lemma}{#1}}

\def\Df#1{\exdeclare{Definition}{#1}}

\def\Remark#1\par{\bigskip \is{Remark}: \ #1\par\medbreak}

\def\exdeclare#1#2#3\par{\begin{#1}\label{#2}{\rm #3}\end{#1}\rm}
\def\stdeclare#1#2#3\par{\begin{#1}\label{#2}{#3}\end{#1}}

\def\stwrite#1#2{#1 \stn{#2}}  
\def\stn{\ref}

\def\Labl#1{{\it#1\/}:\enspace}
\def\Label#1{\medbreak\noindent\Labl{#1}}
\def\Labelpar #1.#2\par{\Label{#1}#2\par\medbreak}

\def\Ack #1\par{\Labelpar Acknowledgments.#1\par}

\def\andq{\roq{and}}

\def\Proof{\Label{Proof}}
\def\qed{\hfill\parfillskip=0pt {\it Q.E.D.}\par\parfillskip=0pt plus 1fil\medbreak}

\def\Eq#1$$#2$${\StEq#1  \EnEq{#2}}
\def\StEq#1 #2\EnEq#3{\begin{equation}\label{#1eq*} #3\end{equation}}

\newdimen\eqjot \eqjot = 1\jot
\def\openupeq{\openup \the\eqjot}
\newdimen\tablejot \tablejot = 1\jot

\def\openuptable{\openup \the\tablejot}
\def\addtab#1={#1\;&=}

\def\eqrefz#1{\ref{#1eq*}}
\def\eq#1{{\rm (\eqrefz{#1})}}
\def\eqe{equation~\eq}   
\def\eqs#1#2{\eq{#1},~\eq{#2}}

\def\eqc#1#2{\iflongmode\eq{#1},~\eq{#2}\else{\rm (\eqrefz{#1}, \eqrefz{#2})}\fi}

\def\eqr#1#2{\iflongmode\eq{#1}--\eq{#2}\else{\rm (\eqrefz{#1}--\eqrefz{#2})}\fi}
\def\eqf{formula~\eq}

\def\qaeq#1#2{{\def\\{&}\vcenter{\openupeq\halign{$\displaystyle
   ##\hfil$&&\hskip#1pt$\displaystyle##\hfil$\cr #2\cr}}}}
 
\def\req{\qaeq{30}}  
\def\qeq{\qaeq{20}}

\def\ezeq#1#2#3{{\def\\{\cr#1}\vcenter{\openupeq \halign{$\displaystyle 
   \hfil##$&$\displaystyle##\hfil$&&\hskip#2pt$\displaystyle##\hfil$\cr#1#3\cr}}}}
\def\eaeq{\ezeq\addtab}

\def\eeq{\eaeq{20}}  

\def\caeq#1#2{{\def\\{\cr}\vcenter{\openupeq \halign
    {$\hfil\displaystyle ##\hfil$&&$\hfil\hskip#1pt\displaystyle ##\hfil$\cr#2\cr}}}}
\def\ceq{\caeq{20}}  

\def\iaeq#1#2#3{{\def\\{\cr\hskip#1pt}\vcenter{\openupeq\halign{$\displaystyle
   ##\hfil$&&\hskip#2pt$\displaystyle##\hfil$\cr #3\cr}}}}
\def\ibeq#1{\iaeq{#1}{#1}}

\def\crh#1{{}\cr\hskip#1pt{}}
   \def\crhQ{\crh{40}}

\def\qarray#1{\null\,\vcenter{\normalbaselines\mathsurround=0pt
    \ialign{\hfil$##$\hfil&&\quad\hfil$##$\hfil\crcr\mathstrut\crcr\noalign{\kern-\baselineskip}
      #1\crcr\mathstrut\crcr\noalign{\kern-\baselineskip}}}\,}
\def\parray#1{\left(\qarray{#1}\right)}
\def\omatrix#1{{\def\\{\cr}\parray{#1\cr}}}

\def\where{\roq{where}}
\def\withq{\roq{with}}
\def\forq{\roq{for}}
 
\def\dsty{\displaystyle}
\def\rg#1#2{#1=1,\ldots,#2} 
\def\ps#1{^{(#1)}}

\newif\ifnames \namesfalse
\def\namez#1;{\namesfalse\namezz#1,;,} 
\def\namezz#1,#2,#3;{\ifx,#3,\ifnames \ and\fi\else \ifnames,\fi\fi #2 \ignorespaces#1\ifx,#3,\else \namestrue\namezz #3;\fi}

\def\key#1 #2\par{\bibitem{#1} #2\par}

\def\book#1;#2;#3\par{\namez#1;{\it #2},#3.}
\def\paper#1;#2;#3; #4(#5)#6\par{\namez#1;{\it #2},{#3}\if a#4, to appear.\else
     { \bf #4}(#5)\if n#6.\else,#6.\fi\fi}
\def\inbook#1;#2;#3;#4\par{\namez#1;#2, {\sl in\/}:{\it#3},#4.}
\def\appx#1;#2;#3;#4\par{\namez#1;#2, {\sl appendix in\/}:{\it #3},#4.\par }
\def\thesis#1;#2;#3\par{\namez#1;{\it #2}, Ph.D.~Thesis,#3.}
\def\uthesis#1;#2;#3\par{\namez#1;{\it #2}, Ph.D.~Thesis, University of Minnesota,#3.}
\def\preprint#1;#2;#3\par{\namez#1;#2, preprint,#3.}
\def\upreprint#1;#2;#3\par{\namez#1;#2, preprint, University of Minnesota,#3.}
\def\prepare#1;#2;#3\par{\namez#1;#2, in preparation.}
\def\personal#1;#2\par{\namez#1; personal communication,#2.}
\def\submit#1;#2;#3\par{\namez#1;#2, submitted.}
\def\other#1\par{#1.}

\def\nosemiz#1;{#1}

\def\psk{\hskip2pt}

\def\rf#1{{\rm \rfzzz#1;\rfzzz}}
\def\rfzzz#1;#2\rfzzz{\ifx;#2;\cite{#1}\else\cite[\nosemiz\ignorespaces#2]{#1}\fi}
\def\ie{i.e., }
\def\cf{cf.~} 

\def\is#1{{\em #1}} 
\def\bo#1{{\bf #1}}

\def\ro#1{{\rm #1}} \def\rbox#1{\hbox{\rm #1}}
\def\roh{\rbox}

\def\rox#1{\quad \rbox{#1}\quad }
\def\roxx#1{\ \quad \rbox{#1}\quad \ }

\def\roq#1{\qquad \rbox{#1}\qquad }

\def\slbox#1{\hbox{\sl #1}}

\def\soq#1{\qquad \slbox{#1}\qquad }

\def\R{\mathbb R}  
\def\Rx{\R\mkern1mu^}  \def\Cx{\C\mkern1mu^}
\def\Rn{\Rx n}  \def\Rk{\Rx k}

\def\ldotsx{\ \ldots\ } \def\cdotsx{\ \cdots\ }

\def\subs #1#2{#1_1,\ldots,#1_{#2}}
\def\psubs #1#2{(#1_1,\ldots,#1_{#2})}
\def\rvector#1{(\:#1\:)}
\def\tvector#1{\rvector{#1}^T}
\def\ptvector#1{\pa{#1}^T}
\def\rvectors#1#2{\rvector{#1_1,#1_2,\ldots,#1_{#2}}}
\def\tvectors#1#2{\rvectors{#1}{#2}^T}
\def\ip#1#2{\langle \,{#1}\,;{#2}\,\rangle}
\def\ipdpp#1#2{(#1)\cdot (#2)}
\def\Ip#1#2{\left \langle \,{#1}\,;\,{#2}\,\right \rangle}
\def\nnorm#1{ \|\,{#1}\,\|} 
\def\mx#1{#1\times #1} \def\mxp#1{(#1)\times(#1)} 

\def\f#1{{\textstyle \frac1{#1}}} 
\def\fra#1{\frac{1}{#1}}

\def\Sumu#1{\sum _{#1}\>}
\def\Summ#1#2#3{\sum _{#1\:=\:#2}^{#3}\>}
\def\Sum#1#2{\Summ{#1}1{#2}}  

\def\operator#1{\expandafter\def\csname#1\endcsname{\mathop{\rm #1}\nolimits}}
\operator{rank}

                       \let\pa\parens
             
        \let\bpa\bparens

  \let\bc\braces

\def\ostrut#1#2{\hbox{\vrule height #1pt depth #2pt width 0pt}}
\def\sstrut{\ostrut0} \def\ustrut#1{\ostrut{#1}0}
\def\msk#1{\mskip #1 mu} \def\:{\mskip2mu}

\def\upth{-{th} }  \def\upst{-{st} }  
\def\xth #1{$#1$\upth}  \def\xst #1{$#1$\upst}
  
\def\ith{\xth i}    \def\kth{\xth k}
  \def\nth{\xth n}    

\def\Chat{\widehat C{}}

\def\Mhat{\widehat M{}}

\def\Ctilde{\widetilde C{}}

\font\tensmc=cmcsc10 scaled 1200  \def\smc{\tensmc}
\def\Mathematica{{\smc Mathematica}} \def\Maple{{\smc Maple}} 

\def\p{\bo p} \def\x{\bo x} \def\y{\bo y}  \def\z{\bo z}
\def\xh{\widehat {\x}}

\def\Ring{{\cal R}} \def\Quo{{\cal Q}} \def\Alg{{\cal A}}

\def\o{\bo 0} \def\szerO{\setb\o} 
\def\Pt{\widetilde P}

\def\Rd{\Rx d}

\def\Cx{{\cal K}} \def\Cn{\Cx\ps n} \def\Cno{\Cx_0\ps n}
\def\H{{\cal H}}

\let\ip\ipdpp
\def\esp#1{e_{#1}}

\long\def\ignore#1{}

\def\setb#1{\{#1\}}
\def\set#1#2{\{#1, #2\}}
\def\setm#1#2{\{#1\mid #2\}}

\newcommand{\abs}[1]{\left| #1 \right|}
\newcommand{\ive}[1]{\left[ #1 \right]}
\ihead{The $n$ Body Matrix and its Determinant}
\ohead{\today}

\begin{document}

\title{The $n$ Body Matrix and its Determinant}
\date{\today}
\author{Darij Grinberg$^1$ and Peter J. Olver$^2$\\\\
\footnotesize 1. School of Mathematics, University of Minnesota, Minneapolis, MN 55455\\ \footnotesize{\texttt{darijgrinberg@gmail.com}}\qquad  \footnotesize{\url{http://www.cip.ifi.lmu.de/~grinberg/}} \\
\footnotesize 2. School of Mathematics, University of Minnesota, Minneapolis, MN 55455\\ \footnotesize{\texttt{olver@umn.edu}}\qquad  \footnotesize{\url{http://www.math.umn.edu/~olver}}
}

\maketitle
%%%%%%%%%%%%%%%%%%%%%%%%%%%%%

\vskip-.52in
\Abstract  The primary purpose of this note is to prove two recent conjectures concerning the $n$ body matrix that arose in recent papers of Escobar-Ruiz, Miller, and Turbiner on the classical and quantum $n$ body problem in $d$-dimensional space.  First, whenever the positions of the masses are in a nonsingular configuration, meaning that they do not lie on an affine subspace of dimension $\leq n-2$, the $n$ body matrix is positive definite and, hence, defines a Riemannian metric on the space coordinatized by their interpoint distances.  Second, its determinant can be factored into the product of the order $n$ Cayley--Menger determinant and a mass-dependent factor that is also of one sign on all nonsingular mass configurations.  The factorization of the $n$ body determinant is shown to be a special case of an intriguing general result proving the factorization of determinants of a certain form.

\vfill\eject

\Section i The $n$ Body Matrix.

The \is{$n$ body problem}, meaning the motion of $n$ point masses (or point charges) in $d$-dimensional space under the influence of a potential that depends solely on pairwise distances, has a venerable history,
capturing the attention of many prominent mathematicians, including Euler, Lagrange, Dirichlet, Poincar\'e, Sundman, etc., \rf{SM, Whittaker}.  The corresponding quantum mechanical system, obtained by quantizing the classical Hamiltonian to form a Schr\"odinger operator, has been of pre-eminent interest since the dawn of quantum mechanics, \rf{LLqm}.  

In three recent papers, \rf{MTE,TME,TMEd}, Escobar-Ruiz, Miller, and Turbiner  made the following remarkable observation.  Once the center of mass coordinates have been separated out, the quantum $n$ body Schr\"odinger operator decomposes into a ``radial'' component that depends only upon the distances between the masses plus an ``angular'' component that involves the remaining coordinates and annihilates all functions of the interpoint distances.  
Moreover, the radial component is gauge equivalent to a second order differential operator which, as we will prove, is the Laplace--Beltrami operator on a certain Riemannian manifold coordinatized by the interpoint distances, whose geometry is as yet not well understood.  This decomposition allows one to separate out the ``radial'' eigenstates that depend only upon the interpoint distances from the more general eigenstates that also involve the angular coordinates.  A similar
%, previously unknown  (except in the very special case $n=3$ dating back to Lagrange),  
separation arises in the classical $n$ body problem through the process of ``dequantization'', \ie reversion to the classical limit.  

The primary goal of this paper is to prove several fundamental conjectures that were made in \rf{MTE} concerning the algebraic structure of the underlying $n$ body radial metric tensor.  
To be precise, suppose the point masses $\subs mn$ occupy positions\fnote{We work with column vectors in $\Rd$ throughout.} 
$$\req{\p_i = \tvector{p_i^1,\ldotsx,p_i^d}\in \Rd,\\\rg in.}$$
We will refer to this as a \is{point configuration} consisting of $n$ points (representing the positions of the $n$ masses) in $d$-dimensional space.  To streamline the presentation, we will allow one or more of the points to coincide, \ie permit ``collisions'' of the masses.  The subsequent formulas will slightly simplify if we express them in terms of the inverse masses
\Eq{ai}
$$\req{\alpha _i = \fra{m_i},\\\rg in.}$$

\Df{sing} The point configuration $\subs \p n$ will be called \is{singular} if and only if its elements lie on a common affine subspace of dimension $\leq n-2$.  

Thus, three points are singular if they are collinear; four points are singular if they are coplanar or collinear; etc.  Note that non-singularity requires that the underlying space be of sufficiently high dimension, namely $d \geq n-1$.

Using the usual dot product and Euclidean norm\fnote{It would be an interesting project to extend our results to other norms and, more generally, to Riemannian manifolds.}, let
\Eq{rij}
$$\req{r_{ij} = r_{ji} = \nnorm{\p_i-\p_j} = \sqrt{\ip{\p_i-\p_j}{\p_i-\p_j}}\,,\\ i \ne j,}$$
denote the ${n \choose 2} = \f2\: n(n-1)$ interpoint distances $r = (\ldotsx r_{ij} \ldotsx)$.  On occasion, formulas may also involve $r_{ii} = 0$. The subset  traced out by the $r_{ij}$'s corresponding to all $n$ point configurations (for any, or, equivalently, sufficiently large $d$) will be called the \is{Euclidean distance cone}, which forms a closed conical subset of the nonnegative orthant in $\R^{n \choose 2}$, and denoted\fnote{In general, $\Rk_{\geq\, 0} = \setm{\x = \tvectors xk}{x_i \geq 0, \ \rg ik}$ will denote the nonnegative orthant and $\Rk_{> \,0} = \setm{\x = \tvectors xk}{x_i >0, \ \rg ik}$ the positive orthant of $\Rk$.}
\Eq{RSn}
$$\Cn \subset \R^{n \choose 2}_{\geq\, 0}.$$
We will explicitly characterize $\Cn$ in \th{SYH} below.  As we will see, its boundary $\partial \Cn$ consists of the interpoint distances determined by singular point configurations while its interior, denoted $\Cno\sstrut8$, will correspond to the nonsingular configurations.

%The \is{$n$ body matrix} $B = B\ps n$  is defined as the $\mx{\f2\:n(n-1)}\sstrut6$ matrix whose rows and columns are indexed by unordered pairs $\set ij=\set ji$ of distinct integers $1 \leq i < j \leq n$.  Its diagonal entries are
%$$b_{\set ij,\set ij} = 2\, \frac{m_i+m_j}{m_im_j}\: r_{ij}^2 = \Paz{\frac2{m_i}+\frac2{m_j}} \ip{\p_i-\p_j}{\p_i-\p_j},$$
%while its off diagonal entries are
%$$\eeq{b_{\set ij,\set ik} = \fra{m_i}\pa{r_{ij}^2 + r_{ik}^2 - r_{jk}^2} = \frac2{m_i}\,\ip{\p_i-\p_j}{\p_i-\p_k}, & i,j,k \ \roh{distinct},\\
%b_{\set ij,\set kl} = 0,& i,j,k,l\ \roh{distinct}.}$$
%Then 
%\Eq{Delta}
%$$\Delta\ps n = \det B\ps n$$ 
%will denote the \is{$n$ body determinant}.

The \is{$n$ body matrix} $B\ps n = B\ps n(\alpha,r)$ defined in \rf{MTE} is the $\mx{{n \choose 2}}\sstrut6$ matrix whose rows and columns are indexed by unordered pairs $\set ij=\set ji\sstrut6$ of distinct integers $1 \leq i < j \leq n$.  Its diagonal entries are
\Eq{Bd}
$$b_{\set ij,\set ij} = 2\pa{\alpha _i + \alpha _j}\: r_{ij}^2 = 2\pa{\alpha _i + \alpha _j}\, \ip{\p_i-\p_j}{\p_i-\p_j},$$
while its off diagonal entries are
\Eq{Bo}
$$\eeq{b_{\set ij,\set ik} = \alpha _i\pa{r_{ij}^2 + r_{ik}^2 - r_{jk}^2} = 2\:\alpha _i\,\ip{\p_i-\p_j}{\p_i-\p_k}, & i,j,k \ \roh{distinct},\\
b_{\set ij,\set kl} = 0,& i,j,k,l\ \roh{distinct}.}$$
For example, the $3$ body matrix is
\Eq{B3} 
$$\eeq{B\ps3 = \omatrix{2\: (\alpha _1 + \alpha _2) r_{12}^2&\alpha _1(r_{12}^2 + r_{13}^2 - r_{23}^2)&\alpha _2(r_{12}^2 + r_{23}^2 - r_{13}^2)\sstrut6\\
\alpha _1(r_{12}^2 + r_{13}^2 - r_{23}^2)&2\: (\alpha _1 + \alpha _3) r_{13}^2&\alpha _3(r_{13}^2 + r_{23}^2 - r_{12}^2)\sstrut6\\
\alpha _2(r_{12}^2 + r_{23}^2 - r_{13}^2)&\alpha _3(r_{13}^2 + r_{23}^2 - r_{12}^2)&2\: ( \alpha _2 + \alpha _3) r_{23}^2}},$$
where the rows and the columns are ordered as follows: $\set 12, \set 13, \set 23$.
Our first main result concerns its positive definiteness.

\Th{nbpos} The $n$ body matrix $B\ps n(\alpha,r)$ is positive semi-definite for $\alpha\in \Rn_{\geq \,0}$ and $r \in \Cn$, and is positive definite if and only if $\alpha\in \Rn_{>\, 0}$ and $r \in \Cno\sstrut8$, i.e., the masses are all positive and situated in a non-singular point configuration.
%the point configuration $\subs \p n$ is non-singular.

Consequently, for each fixed value of the mass parameters $\alpha \in \Rn_{> 0}$, the $n$ body matrix $B\ps n(\alpha ,r)$  defines a Riemannian metric on the interior of the Euclidean distance cone.  This result implies that the radial component of the quantum $n$ body Schr\"odinger operator derived in \rf{MTE} is the corresponding elliptic Laplace--Beltrami operator on the Riemannian manifold $\Cno$.  The underlying Riemannian geometry of $\Cno\sstrut8$ is not well understood.

%\Remark Since the masses lie at distinct locations, the interpoint distances \eq{rij} are all positive, $r_{ij} > 0$, and are further constrained by the triangle inequalities.  Thus the space they coordinatize is strictly contained in the positive orthant of $\R^{n(n-1)/2}$. 

The determinant of the $n$ body matrix 
\Eq{Delta}
$$\Delta\ps n= \Delta\ps n(\alpha,r) = \det B\ps n(\alpha,r)$$ 
will be called the \is{$n$ body determinant}.
For example, a short computation based on \eq{B3} shows that the {$3$ body determinant} can be written in the following factored form:
\Eq{B3d}
$$\ibeq{130}{\Delta\ps3 = \det B\ps3 = -2\:(\alpha _1 \alpha _2 + \alpha _1 \alpha _3 + \alpha _2 \alpha _3)\,(\alpha _3\:r_{12}^2 + \alpha _2\:r_{13}^2 +\alpha _1 \:r_{23}^2)\\(r_{12}^4 + r_{13}^4 + r_{23}^4 - 2\:r_{12}^2r_{13}^2 - 2\:r_{12}^2r_{23}^2 - 2\:r_{13}^2r_{23}^2).}$$
Two important things to notice:  ignoring the initial numerical factor, the first factor is the elementary symmetric polynomial of degree $n-1 = 2$ in the mass parameters $\alpha _i = 1/m_i$; further, the final polynomial factor is purely \is{geometric}, meaning that it is independent of the mass parameters, and so only depends on the configuration of their locations through their interpoint distances.  Positive definiteness of $B\ps3$ implies $\Delta\ps3 > 0$ for nonsingular (\ie non-collinear) configurations.  In view of the sign of the initial numerical factor, this clearly implies the final geometrical factor is strictly negative on such configurations, a fact that is not immediately evident and in fact requires that the $r_{ij}$'s be interpoint distances satisfying the triangle inequalities, i.e., $r = (r_{12},r_{13},r_{23}) \in \Cx_0\ps 3$; indeed, this factor is obviously positive for some non-geometrical values of the $r_{ij}$'s. 
Similar factorizations were found in \rf{MTE} for the cases $n=2,3,4$, and, in the case of equal masses, $n=5,6$, via symbolic calculations using both \Mathematica\ and \Maple.

The geometrical factors that appear in each of these computed factorizations are, in fact, well known, and equal to the \is{Cayley--Menger determinant} of order $n$, a quantity that arises in the very first paper of Arthur Cayley, \rf{Cayleyp1}, written before he turned 20 and, apparently,  inspired by reading Lagrange and Laplace!  In this paper, Cayley uses the relatively new theorem that the determinant (a quantity he calls ``tolerably known'') of the product of two matrices is the product of their determinants in order to solve the problem of finding the algebraic condition (or syzygy) relating the interpoint distances among singular configurations of $5$ points in a three-dimensional space, as well as $4$ coplanar points and $3$ collinear points, each of which is characterized by the vanishing of their respective Cayley--Menger determinant.  A century later, in the hands of Karl Menger, this determinantal quantity laid the foundation of the active contemporary field of distance geometry, \rf{Blumenthal,LiLa}; see also \rf{Oji} for further results and extensions to other geometries.

%There are two ways of constructing the Cayley--Menger determinants. 
The \is{order $n$ Cayley--Menger matrix} is the symmetric  matrix
\Eq{lCM}
$$
%C_R = 
C\ps n= C\ps n(r) = \omatrix{0&r_{12}^2&r_{13}^2&\ldots&r_{1n}^2&1\sstrut6\\
r_{12}^2&0&r_{23}^2&\ldots&r_{2n}^2&1\sstrut6\\
r_{13}^2&r_{23}^2&0&\ldots&r_{3n}^2&1\\
\vdots&\vdots&\vdots&\ddots&\vdots\\
%r_{1,n-1}^2&r_{2,n-1}^2&r_{3,n-1}^2&\ldots&r_{n-1,n}&1\\
r_{1n}^2&r_{2n}^2&r_{3n}^2&\ldots&0&1\sstrut8\\
1&1&1&\ldots&1&0}$$
of size  $\mxp{n+1}$ involving the same interpoint distances \eq{rij}.  The Cayley--Menger matrix is a bordered version of the \is{Euclidean} (\is{squared\/}) \is{distance matrix}\fnote{Even though its entries are the \is{squared} Euclidean distances, the potentially confusing terminology ``Euclidean distance matrix'' is almost universal throughout the literature, and plays the preponderant role in all applications.}
\Eq{EDM}
$$
D\ps n =D\ps n(r) = \omatrix{0&r_{12}^2&r_{13}^2&\ldots&r_{1n}^2\sstrut6\\
r_{12}^2&0&r_{23}^2&\ldots&r_{2n}^2\sstrut6\\
r_{13}^2&r_{23}^2&0&\ldots&r_{3n}^2\\
\vdots&\vdots&\vdots&\ddots&\vdots\\
%r_{1,n-1}^2&r_{2,n-1}^2&r_{3,n-1}^2&\ldots&r_{n-1,n}\\
r_{1n}^2&r_{2n}^2&r_{3n}^2&\ldots&0\sstrut8}$$
of importance in a wide range of applications, including statistics, crystallography, protein structure, machine learning, sensor networks, acoustics, psychometrics, and elsewhere, \rf{DPRV,LiLa}.

The \is{order $n$ Cayley--Menger determinant} is defined as the determinant of the Cayley--Menger matrix:
\Eq{CMd} 
$$\delta\ps n(r)  = \det C\ps n(r).$$
For example, when $n=3$,
\Eq{lCM3}
$$\ceq{C\ps 3 = \omatrix{0&r_{12}^2&r_{13}^2&1\sstrut6\\
r_{12}^2&0&r_{23}^2&1\sstrut6\\
r_{13}^2&r_{23}^2&0&1\sstrut8\\
1&1&1&0},\\
\delta\ps3  = \det C\ps3 = r_{12}^4 + r_{13}^4 + r_{23}^4 - 2\:r_{12}^2r_{13}^2 - 2\:r_{12}^2r_{23}^2 - 2\:r_{13}^2r_{23}^2,}$$
which coincides with the geometric polynomial factor in \eq{B3d}.
Keep in mind that both the $n$ body and Cayley--Menger determinants are homogeneous polynomials in the squared distances $r_{ij}^2$. 
 The general form of Cayley's result can be stated as follows.

\Th{CMdet} A collection of (squared) interpoint distances $r \in \Cn$ comes from a singular point configuration, so  $r \in \partial \Cn$,  if and only if the corresponding Cayley--Menger determinant vanishes\/\ro: $\delta\ps n(r)  =0$.

In other words, the boundary of the Euclidean distance cone is contained in the subvariety determined by the vanishing of a single polynomial --- the Cayley--Menger determinant.  Thus, \th{nbpos} implies that the $n$ body determinant, and hence the underlying metric, degenerates if and only if the Cayley--Menger determinant vanishes, and hence the masses are positioned on a lower dimensional affine subspace.  See below for a modern version of Cayley's original proof.

\Remark When $n=3$, the Cayley--Menger determinant \eq{lCM3} factorizes: 
%when written in terms of the $r_{ij}$,
\Eq{Heron}
$$\delta\ps3(r) = - (r_{12} + r_{13} + r_{23})(-\:r_{12} + r_{13} + r_{23})(r_{12} - r_{13} + r_{23})(r_{12} + r_{13} - r_{23}),$$
which, except for the sign and a factor of $\frac1{16}$, is \is{Heron's formula} for the squared area of a triangle. On the other hand, when $n \geq 4$, the Cayley--Menger determinant is an irreducible polynomial in the distance variables $r_{ij}$; see \rf{DASom,HHNS}, keeping in mind that their $n$ is our $n-1$.

We further note that the relation to Heron's formula is no accident.
Indeed, if $r_{ij}$ are the interpoint distances between $n$ points $\subs \p n\in \Rd$,
then, by a theorem of Menger, \rf{MengerEG} --- see also (**) in Section 40 of \rf{Blumenthal} --- the Cayley--Menger determinant $\delta\ps n(r)$ equals $\left(-1\right)^n 2^{n-1} \left(n-1\right)!^2$ times the squared volume of the $n$-simplex formed by these points.

Based on their above-mentioned symbolic calculations, Miller, Turbiner, and Escobar--Ruiz, \rf{MTE}, conjectured the following result.

%observed that, for small values of $n$, the $n$ body determinant can be factored, with one the factors being the order $n$ Cayley--Menger determinant.  They thus conjectured the following result.

\Th{factor} The $n$ body determinant factors,
\Eq{factor}
$$\Delta\ps n(\alpha,r) = \esp{n-1}(\alpha)\, \delta \ps n(r)\,\sigma \ps n(\alpha,r),$$
into the product of the elementary symmetric polynomial $\esp{n-1}(\alpha)$ of order $n-1$ in the mass parameters $\alpha = \psubs \alpha n$ times the Cayley--Menger determinant $\delta \ps n(r)$ of order $n$ depending only on the interpoint distances $r = (\ldotsx r_{ij}\ldotsx )$ times a polynomial $\sigma \ps n$ that depends upon both the $\alpha _i$ and the  $r_{ij}$.

%In this note, we establish the validity of \th{factor} for all $n$. It is interesting to note that this factorization holds for any values of the mass parameters $\alpha _i = 1/m_i$, even though the Cayley--Menger factor is purely geometrical, \ie only depends upon the distances between the mass locations.
 
Unfortunately, our proof of \th{factor} is purely existential; it does not yield an independent formula for the non-geometrical factor, other than the obvious $\sigma \ps n(\alpha,r) = \Delta\ps n(\alpha,r)/\pa{\esp{n-1}(\alpha)\, \delta \ps n(r)}$.  Thus, the problem of characterizing and understanding the non-geometric factor $\sigma\ps n $ remains open, although interesting formulas involving geometric quantities --- a polynomial involving sums of the squared volumes of the subsimplices determined by the point configuration weighted by the mass parameters --- are known when $n$ is small, \rf{MTE}, who conjecture that this will hold in general.  
%Nor does the proof give any insight into the geometry of the Riemannian manifold $\Cno$ whose metric tensor is prescribed by the $n$ body matrix $B\ps n(\alpha ,r)$ for fixed $\alpha \in \Rn_{>\,0}$.  

We shall, in fact, prove \th{factor} as a special case of a much more general determinantal factorization, \th{Wfactor}, which replaces the squared distances $r_{ij}^2$ by $n^2$ arbitrary elements $s_{i,j}$, not necessarily satisfying either $s_{i,j} = s_{j,i}$ or $s_{i,i} = 0$. We shall also generalize the dependence on the inverse mass parameters $ \alpha _i$ using the following elementary observation.  

\Lm{CA} Given the parameters $\subs \alpha n$,
consider the following $(n+1) \times (n+1)$ matrix
\Eq{CA}
$$
C_{A}=\omatrix{\alpha_{1} & 0 & \cdots & 0 & 1\\
0 & \alpha_{2} & \cdots & 0 & 1\\
\vdots & \vdots & \ddots & \vdots & \vdots\\
0 & 0 & \cdots & \alpha_{n} & 1\\
1 & 1 & \cdots & 1 & 0}.
$$
Then,
\Eq{dCA}
$$\det  C_{A}   =-\:e_{n-1}\left(  \alpha\right)  .$$

\smallskip

To establish this formula, one can simply expand the determinant along its last row.  Thus, the two initial factors in the $n$ body determinant factorization formula \eq{factor} are both realized by determinants of $(n+1) \times (n+1)$-matrices whose final row and column are of a very particular form.
In our further generalization of the $n$ body determinant factorization formula \eq{factor}, we will replace the upper left $n \times n$ block in \eq{lCM}  by a general matrix depending on $n^2$ arbitrary elements $s_{i,j}$ and the upper left $n \times n$ block in \eq{CA} by a general matrix depending on an additional $n^2$ arbitrary elements $t_{k,l}$.  See below for details.

%\begin{proof}
%Expanding $\det\left(  C_{A}\right)  $ along the last row yields
%\begin{align*}
%\det\left(  C_{A}\right)    & =\sum_{i\in\left[  n\right]  }\left(  -1\right)
%^{i+\left(  n+1\right)  }1\underbrace{\det\left(
%\begin{array}
%[c]{ccccccccc}%
%\alpha_{1} & 0 & 0 & 0 & 0 & 0 & 0 & 0 & 1\\
%0 & \alpha_{2} & 0 & 0 & 0 & 0 & 0 & 0 & 1\\
%\vdots & \vdots & \ddots & \vdots & \vdots & \vdots & \vdots & \vdots &
%\vdots\\
%0 & 0 & 0 & \alpha_{i-1} & 0 & 0 & 0 & 0 & 1\\
%0 & 0 & 0 & 0 & 0 & 0 & 0 & 0 & 1\\
%0 & 0 & 0 & 0 & \alpha_{i+1} & 0 & 0 & 0 & 1\\
%0 & 0 & 0 & 0 & 0 & \alpha_{i+2} & 0 & 0 & 1\\
%0 & 0 & 0 & 0 & 0 & 0 & \ddots & 0 & \vdots\\
%0 & 0 & 0 & 0 & 0 & 0 & 0 & \alpha_{n} & 0
%\end{array}
%\right)  }_{\substack{=\left(  -1\right)  ^{n-i}\alpha_{1}\alpha_{2}%
%\cdots\alpha_{i-1}\alpha_{i+1}\alpha_{i+2}\cdots\alpha_{n}1\\\text{(since
%moving the }i\text{-th row of this matrix}\\\text{to the very bottom
%transforms it into an}\\\text{upper-triangular marix with diagonal}%
%\\\text{entries }\alpha_{1},\alpha_{2},\ldots,\alpha_{i-1},\alpha_{i+1}%
%,\alpha_{i+2},\ldots,\alpha_{n},1\text{)}}}\\
%& =\sum_{i\in\left[  n\right]  }\underbrace{\left(  -1\right)  ^{i+\left(
%n+1\right)  }1\left(  -1\right)  ^{n-i}}_{=-1}\alpha_{1}\alpha_{2}\cdots
%\alpha_{i-1}\alpha_{i+1}\alpha_{i+2}\cdots\alpha_{n}1\\
%& =-\underbrace{\sum_{i\in\left[  n\right]  }\alpha_{1}\alpha_{2}\cdots
%\alpha_{i-1}\alpha_{i+1}\alpha_{i+2}\cdots\alpha_{n}}_{=e_{n-1}\left(
%\alpha\right)  }=-e_{n-1}\left(  \alpha\right)  ,
%\end{align*}
%which proves \eq{dCA}.
%\end{proof}

Combining \ths{nbpos}{factor} allows us to resolve another conjecture in \rf{MTE}, that for nonsingular point configurations, the mass-dependent factor $\sigma \ps n(\alpha, r)$ is of one sign. 
 
\Th{possig} All factors in the $n$ body determinant factorization \eq{factor} are of one sign, namely
\Eq{posdsd}
$$\qeq{\Delta\ps n(\alpha, r) > 0,\\ \esp{n-1}(\alpha) > 0,\\(-1)^n \,\delta \ps n( r)> 0,\\(-1)^n \:\sigma \ps n(\alpha, r) > 0,}$$
 provided the mass parameters $\alpha _i = 1/m_i$ are positive, so $\alpha \in \Rn_{> \, 0}$, and their positions $\p_i$ do not all lie in an affine subspace of dimension $\leq n-2$, so $r \in \Cno$. 

\Proof
Since the determinant of a positive definite matrix is positive,  \rf{Lauritzen,OS},
%(a consequence, \textit{e.g.}, of Theorem 9.11 in \rf{Lauritzen}), 
\th{nbpos} immediately implies the first inequality in \eq{posdsd}.  The positivity of the elementary symmetric polynomial is trivial.  The sign of the Cayley--Menger determinant $\delta \ps n(r)$ on nonsingular point configurations is well known; see \eq{deltageq0} below for a proof.  The final inequality follows immediately from the factorization \eq{factor}.\qed

\Section p Positive Definiteness.

In this section, we present a proof of \th{nbpos}, as well as the  known 
%\th{CMdet} and the positive definiteness of the reduced Euclidean distance matrices, to be introduced next.  
results concerning the vanishing and the sign of the Cayley--Menger determinant.  These results will, modulo the proof of the Factorization \th{factor}, establish the Sign \th{possig}.  We begin by introducing an important collection of matrices that are closely related to the Cayley--Menger matrices.
 
Given a point configuration $\subs \p n \in \Rd$, for each $\rg kn$, consider the $\mxp{n-1}$ matrix $M\ps n_k= M\ps n_k(r)$ with entries 
\Eq{Mij}
$$\eeq{m_{ij} = 2\,\ipdpp{\p_i-\p_k}{\p_j-\p_k} = \nnorm{\p_i-\p_k}^2 + \nnorm{\p_j-\p_k}^2 - \nnorm{\p_i-\p_j}^2 \\= r_{ik}^2 + r_{jk}^2 - r_{ij}^2, \qquad \qquad i,j \ne k,}$$
where the indices $i,j$ run from $1$ to $n$ omitting $k$, and where $r_{ii} = 0$.  
Note that its diagonal entries are $m_{ii} = 2\: r_{ik}^2$. Thus, in particular, $\dsty M\ps n_n(r)$ is explicitly given by
\Eq{CM}
$$\hskip-.3in\omatrix{2\:r_{1n}^2&r_{1n}^2 + r_{2n}^2 - r_{12}^2&r_{1n}^2 + r_{3n}^2 - r_{13}^2&\hskip-.1in\ldots&r_{1n}^2 + r_{n-1,n}^2 - r_{1,n-1}^2\sstrut6\\
r_{1n}^2 + r_{2n}^2 - r_{12}^2&2\:r_{2n}^2&r_{2n}^2 + r_{3n}^2 - r_{23}^2&\hskip-.1in\ldots&r_{2n}^2 + r_{n-1,n}^2 - r_{2,n-1}^2\sstrut6\\
r_{1n}^2 + r_{3n}^2 - r_{13}^2&r_{2n}^2 + r_{3n}^2 - r_{23}^2&2\:r_{3n}^2&\hskip-.1in\ldots&r_{3n}^2 + r_{n-1,n}^2 - r_{3,n-1}^2\\
\vdots&\vdots&\vdots&\hskip-.1in\ddots&\vdots\\
r_{1n}^2 + r_{n-1,n}^2 - r_{1,n-1}^2&r_{2n}^2 + r_{n-1,n}^2 - r_{2,n-1}^2&r_{3n}^2 + r_{n-1,n}^2 - r_{3,n-1}^2&\hskip-.1in\ldots&2\:r_{n-1,n}^2},\sstrut{40}$$
with evident modifications for the general case $M\ps n_k(r)$.
For example, when $n=3$, 
\Eq{CM3}
$$\eeq{M\ps3_1 = \omatrix{2\:r_{12}^2&r_{12}^2 + r_{13}^2 - r_{23}^2\sstrut6\\
r_{12}^2 + r_{13}^2 - r_{23}^2&2\:r_{13}^2},\\
M\ps3_2 = \omatrix{2\:r_{12}^2&r_{12}^2 + r_{23}^2 - r_{13}^2\sstrut6\\
r_{12}^2 + r_{23}^2 - r_{13}^2&2\:r_{23}^2},\\
M\ps3_3 = \omatrix{2\:r_{13}^2&r_{13}^2 + r_{23}^2 - r_{12}^2\sstrut6\\
r_{13}^2 + r_{23}^2 - r_{12}^2&2\:r_{23}^2}
% \\= 2\omatrix{\ip{\p_1-\p_3}{\p_1-\p_3}&\ip{\p_1-\p_3}{\p_2-\p_3}\\
%\ip{\p_1-\p_3}{\p_2-\p_3}&\ip{\p_2-\p_3}{\p_2-\p_3}}
.}$$
The associated quadratic forms 
\Eq{qks}
$$\qeq{q_k(\x) = \x_k^T M\ps n_k \x_k = \Sumu{i,j\ne k} \bpa{r_{ik}^2 + r_{jk}^2 - r_{ij}^2}x_ix_j,}$$
where $\x_k$ is obtained from $\x = \tvectors xn$ by omitting the entry $x_k$, first appear in a paper by Fr\'echet, \rf{Frechet}, in the special case $n=4$, and, in general, in a supplement written by Schoenberg, \rf{Schoenberg}, who uses them to prove the fundamental \th{SYH} below.  Further, as shown by Schoenberg, the quadratic form \eq{qks} equals the negative of the quadratic form 
\Eq{qedm}
$$\qeq{q(\x) = \x^T D\ps n \x = \Sum{i,j}n r_{ij}^2x_ix_j}$$
associated with the  Euclidean distance matrix \eq{EDM} when restricted to the hyperplane 
$$\H = \setm{\x\in\Rn}{x_1 + \cdotsx + x_n = 0},$$
so that  
\Eq{qqk}
$$q(\x) = -\:q_k(\x) \rox{whenever}  x_1 + \cdotsx + x_n = 0, \qquad \rg kn.$$

The matrices $M\ps n_k$ explicitly appear in a slightly later but independently written paper by Gale Young and Alston Householder, \rf{YoHo}, who were motivated by questions arising in psychometrics; their paper contains a much simplified proof of Schoenberg's \th{SYH} characterizing the Euclidean distance cone, as well as establishing their explicit connection with the Cayley--Menger matrices; see \th{CMdk-prop} below. Although they appear often in the distance geometry literature, we have been unable to find an actual name for them.  Schoenberg's quadratic form identity \eq{qqk} suggests calling them \is{reduced Euclidean distance matrices}; alternatively, one could name them after Fr\'echet, Schoenberg,  Young, and Householder, giving the appealing (slightly permuted) acronym FYSH matrix.  However, in view of later developments, we choose to follow the former more descriptive nomenclature here.

The following result shows that the reduced Euclidean distance determinants of order $n$ (i.e.~the determinants of the reduced Euclidean distance matrices) are independent of $k$ and, up to sign, are equal to the order $n$ Cayley--Menger determinant.

\Th{CMdk-prop}
The {Cayley--Menger determinant} coincides, up to sign, with every reduced Euclidean distance determinant of the same order:
\Eq{CMdk} 
$$\delta\ps n(r)  = (-1)^n\,\det M\ps n_k(r),$$
for any value of $\rg kn$.

\Proof
Our proof follows \rf{YoHo}.
Let us concentrate on the case $k=n$, noting that all formulas are invariant under permutations of the points and, hence, it suffices to establish this particular case.  We perform the following elementary row and column operations on the  Cayley--Menger matrix $C\ps n$, \cf \eq{lCM}, that do not affect its determinant.  We subtract its \nth row from the first through \xst{(n-1)}rows, and then subtract its \nth column, which has not changed, from the resulting first through \xst{(n-1)}columns.  The result is the $\left(n+1\right)\times\left(n+1\right)$-matrix\fnote{The bold face $\o$'s indicate row or column vectors of $0$'s, while the ordinary $0$'s are scalars.}
$$\Ctilde\ps n = \omatrix{\dsty -\:M\ps n_n&*&\o\\*&0&1\\\o&1&0} ,$$
where the upper left $\mxp{n-1}$ block is $\dsty -\:M\ps n_n$, the \nth row and column of $\Ctilde\ps n$ are the same as the \nth row and column of $C\ps n$ (the stars indicate the entries, whose precise values are not needed), and the last row and column have all zeros except for their \nth entry.  We can further subtract suitable multiples of the last row and column from the first $n-1$ rows and columns in order to annihilate their \nth entries, leading to
$$\Chat\ps n = \omatrix{\dsty -\:M\ps n_n&\o&\o\\\o&0&1\\\o&1&0}.$$
It is then easy to see that
$$\delta\ps n = \det C\ps n = \det \Ctilde\ps n = \det \Chat\ps n = (-1)^n \det M\ps n_n.$$
\vglue-20pt
\qed

Now, dropping the ${}\ps n$ superscript and ${}_n$ subscript from here on to avoid cluttering the formulas, the first formula in \eq{Mij} implies that, up to a factor of $2$, the reduced Euclidean distance matrix $M = M\ps n_n$ is a \is{Gram matrix}, \rf{LiLa,OS}, namely
\Eq{MA}
$$M = 2\: A^T \! A,\where
A = \bpa{\p_1-\p_n,\ldotsx,\p_{n-1}-\p_n}$$
is the $d \times (n-1)$ matrix with the indicated columns. In accordance with our choice of nomenclature, we will refer to $A$ as the \is{reduced point configuration matrix} since it agrees with the $d \times n$ \is{point configuration matrix} $P = \psubs \p n$, whose columns are the points or mass locations, when restricted to the hyperplane $\H$:
\Eq{AP}
$$\req{A\:\xh = P\:\x \roxx{when} x_1 + \cdotsx + x_n = 0,\roxx{where} \eeq{\xh = \tvectors x{n-1},\\\x = \tvectors xn .}}$$
%, which we will call the \nth \is{point configuration matrix}.
%, so that its entries are given by dot products:
%$$m_{ij} = 2\,\ipdpp{\p_i-\p_n}{\p_j-\p_n}.$$ 
%The \is{Cayley--Menger determinant} is
%\Eq{CMd} 
%$$\delta  = \det M \roq{or maybe} \det (M/2) = \det(A^T  A).$$
%(The choice of normalization is irrelevant.)
We know that $\delta\ps n  = (-1)^n \det M = 0$ if and only if $\ker M \ne \bc{\o}$, meaning there exists $\o \ne \xh \in \mathbb{R}^{n-1}\ustrut{13}$ such that 
\Eq{Mx0}
$$M\:\xh = \o.$$
Multiplying the left hand side by $\xh^T$ and using \eq{MA}, we find
\Eq{CMpos}
$$\xh^TM\:\xh = 2\msk4 \xh^T \! A^T \!A\,\xh = 2\,\nnorm{A\,\xh}^2 \geq 0 \roq{for all} \xh \in \R^{n-1}.$$
This identity implies that the reduced Euclidean distance matrix $M$ is positive semi-definite, and is positive definite if and only if the reduced point configuration matrix has trivial kernel: $\ker A = \szerO$. Consequently, \eq{Mx0} holds if and only if
\Eq{Ax0}
$$A\,\xh = \o.$$
Since we assumed $\xh\ne \o$, this is equivalent to the linear dependence of the columns of $A$, meaning the vectors $\p_1 - \p_n,\ldots,\p_{n-1} - \p_n$ span a subspace of dimension $\leq n-2$, which requires that $\subs \p n$ lie in an affine subspace of dimension $\leq n-2$, \ie they form a singular point configuration.  We conclude that this occurs if and only if the Cayley--Menger determinant vanishes, which thus establishes Cayley's \th{CMdet}.  
Moreover, positive (semi-)definiteness of $M$ implies non-negativity of its determinant, and hence, by \eq{CMdk},
\Eq{deltageq0}
$$(-1)^n\,\delta\ps n(r) \geq 0, \roq{with equality if and only if} \ker A \ne \szerO,$$
thus establishing the second to last inequality in \eq{posdsd}.  Replacing $\p_n$ by $\p_k$ does not change the argument, and hence we have established the following known result.

\Th{CMpos} Given a point configuration $\subs\p n \subset \Rd\sstrut7$, and any $1 \leq k \leq n$, the corresponding reduced Euclidean distance matrix $M\ps n_k(r)\sstrut8$ depending on the interpoint distances $r_{ij} = \nnorm{\p_i-\p_j}\sstrut9$ is positive semi-definite, and is positive definite if and only if the point configuration is nonsingular.
%, meaning that they do not lie in an affine subspace of dimension $\leq n-2$. 

The condition that the reduced Euclidean distance matrix $M\ps n_k$ be positive semi-definite is, in fact, both necessary and sufficient in order that a given collection of nonnegative numbers $r_{ij} \geq 0$ be the interpoint distances of a bona fide point configuration, i.e. satisfy \eq{rij} for some $\subs \p n \in \Rd$.  This theorem is due to Schoenberg, \rf{Schoenberg}, and, independently, to Young and Householder, \rf{YoHo}, whose simple proof is, for completeness, included here.  This key result, in addition to its foundational role in distance geometry, \rf{LiLa}, directly inspired the powerful and widely used method of contemporary statistical analysis known as \is{multidimensional scaling}, \rf{Davidson}.

\Th{SYH} Given $r = (\ldotsx r_{ij} \ldotsx)\in \R^{n \choose 2}_{\geq\, 0}\sstrut7$, there exists $\subs\p n \subset \Rd$ for some $d >0$ such that $r_{ij} = \nnorm{\p_i-\p_j}$ if and only if for any \ro(and hence all\/\ro) $1 \leq k \leq n$, the corresponding reduced Euclidean distance matrix $M\ps n_k(r)$ is positive semi-definite. The point configuration $\subs\p n$ is unique up to a Euclidean transformation of $\Rd$. The minimal $d$ for which this occurs is given by the rank of $M\ps n_k(r)$.

\Proof
Let us without loss of generality set $k=n$.
In view of \th{CMpos}, we need only prove that, given $r \in \R^{n \choose 2}_{\geq\, 0}$, if the corresponding reduced Euclidean distance matrix $M\ps n_n(r)$ is positive semi-definite, then we can find a point configuration for which the $r_{ij}$ are the interpoint distances.  
Consider its spectral factorization
\Eq{sf}
$$M\ps n_n = Q\: D \:Q^T,$$
in which $Q$ is an $\mxp{n-1}$ orthogonal matrix and $D$ is a diagonal matrix of the same size with the eigenvalues $\lambda _ i$, \ $\rg i{n-1}$, on its diagonal and which, by positive semi-definiteness, are all $\geq 0$; see \rf{OS}  for linear-algebraic details.  Then set
\Eq{AQD}
$$A =  \fra{\sqrt 2}  \,\sqrt D\,Q^T,$$
where $\sqrt D$ is the diagonal matrix whose entries are $\sigma _i = \sqrt{\lambda _ i}\,$. Equation \eq{sf} immediately implies that $A$ satisfies the Gram equation \eq{MA}, and hence defines a corresponding reduced point configuration matrix, namely its \ith column gives $\p_i-\p_n$, where the final point $\p_n$ can be arbitrarily specified. Note also that the (nonzero) $\sigma _i$ are the singular values of the reduced point configuration matrix $A$.  The proof of the remaining statements is left to the reader.\qed 

 \th{SYH} identifies the set of positive semi-definite reduced Euclidean distance matrices with the Euclidean distance cone $\Cn \subset \R^{n \choose 2}$.  When $n=3$, in view of Heron's formula \eq{Heron}, positive semi-definiteness reduces to the triangle inequalities among the three interpoint distances.  However, when $n > 3$, positive semi-definiteness imposes additional constraints on the distances beyond those required by the triangle inequalities among each triple of points in the configuration.
 
 \medskip
 
\is{Remark\/}: As a corollary of Schoenberg's identity \eq{qqk}, one sees that a Euclidean distance matrix $D\ps n = D\ps n(r)$, as in \eq{EDM}, arises from a bona fide point configuration if and only if
\Eq{xDx}
$$
\x^TD\ps n\x \leq 0 \rox{for all} \x = \tvectors xn \in \Rn \rox{such that} x_1 + \cdotsx + x_n = 0.$$

\ignore{
Next, consider the $d \times n$ and  $\pa{d+1} \times n$ matrices with the indicated columns:
\Eq P
$$\req{P = \omatrix {\p_1&\p_2&\ldotsx&\p_n},\\\Pt = \omatrix {\p_1&\p_2&\ldotsx&\p_n\\1&1&\ldotsx&1},}$$
the columns of the latter obtained by appending a $1$ to the column vectors $\p_i\sstrut6$, which is reminiscent of the introduction of projective coordinates.
% --- although we will not pursue this here. 
Subtracting the \nth column of $\Pt\ustrut{11}$ from all the other columns produces the matrix
\Eq Q
$$\Quohat = \omatrix{A&\p_n\\\o&1}.$$
Moreover, if we set
\Eq{xstar}
$$\req{\x = (x_1,\ldotsx, x_n)^T \ne \o,\\\xh = \tvectors x{n-1}\ne \o,}$$
then it is easily seen that 
\Eq{Px0}
$$\Pt\:\x = \o \iff  A\:\xh = \o \andq x_1 + x_2 + \cdotsx + x_n  = 0.$$
Writing out the first equation yields
\Eq{px0}
$$\qeq{x_1\p_1 + x_2\p_2 + \cdotsx + x_n\p_n= \o \withq x_1 + x_2 + \cdotsx + x_n  = 0.}$$
Replacing $x_i$ by $- \Sumu{k\ne i} x_k$, this immediately implies that
\Eq{bpx0}
$$\Sum kn x_k(\p_k-\p_i) = \o \roq{for any} \rg in,$$
which is the analog of \eqe{Ax0} for the reduced Euclidean distance matrix $\dsty M\ps n_i$ based at the point $\p_i$.

To complete the proof of \th{factor}, given $\x=\tvectors xn$ as above, let\break $\z = \x\ps2 \in \Rx{n(n-1)/2}$ denote the vector whose entries, indexed by the same unordered pairs as the $n$ body matrix $B = B\ps n$, are the products of distinct entries of $\x$, so 
\Eq{yij}
$$\req{z_{\set ij} = x_i x_j,\\ i\ne j.}$$
Observe that $\z \ne \o$, since, by the conditions in \eqc{xstar}{Px0}, at least two of the entries of $\x\ne \o$ must be nonzero, and so $\z$ has at least one nonzero entry.
We claim that
\Eq{By0}
$$B\:\z = \o.$$
Indeed, referring back to \eqr{Bd}{Bo}, the entry of the vector $B\:\z$ indexed by  $(i\:j)$ is
\Eq{Byij}
$$\hskip3pt\ibeq{20}{2\:\pa{\alpha _i+\alpha _j}\ip{\p_i-\p_j}{\p_i-\p_j}\: x_ix_j \crhQ + 2\:\alpha _i\,\Sumu{k \ne j} \ip{\p_i-\p_j}{\p_i-\p_k}\: x_ix_k + 2\:\alpha _j\,\Sumu{k \ne i} \ip{\p_i-\p_j}{\p_i-\p_k}\: x_jx_k 
\\= 2\:\alpha _ix_i\Ip{\p_i-\p_j}{\Sum kn x_k(\p_i-\p_k)}  + 2\:\alpha _jx_j\Ip{\p_i-\p_j}{\Sum kn x_k(\p_j-\p_k)} = 0,\hskip-23pt}$$
in view of \eq{bpx0}.
Thus, if the Cayley--Menger determinant vanishes, $\delta = 0$, then there exists $\x\ne \o$ with $x_1 + \cdots + x_n = 0$ such that \eq{Px0} holds.  But this implies that $\z \ne \o$ defined by \eq{yij} satisfies \eq{By0}, which implies the $n+1$ body determinant vanishes, $\Delta\ps n  = \Delta  = 0$.
Restoring the index $n$, we conclude that 
\Eq{Dd0}
$$\Delta\ps n = \det B\ps n = 0 \roq{whenever}\delta\ps n = (-1)^n\,\det M\ps n_n = 0.$$

Finally, according to \rf{DASom} --- see also \rf{HHNS} --- when $n \geq 4$, the Cayley--Menger determinant $\delta\ps n$ is an irreducible polynomial in the variables $r_{ij}$.  Thus, \eq{Dd0} implies that $\delta\ps n$ must be a factor of the $n$ body determinant $\Delta \ps n$, proving the conjectured factorization \eq{factor} in these cases.

On the other hand, when $n=3$, the Cayley--Menger determinant \eq{lCM3} factorizes: 
%when written in terms of the $r_{ij}$,
%\Eq{Heron}
$$\delta\ps3 = (r_{12} + r_{13} + r_{23})(-\:r_{12} + r_{13} + r_{23})(r_{12} - r_{13} + r_{23})(r_{12} + r_{13} - r_{23}),$$
which is \is{Heron's formula} for the squared area of a triangle, \rf{Oji}.  Thus, when $n=2$ or $3$, one needs to perform an easy explicit calculation to establish the factorization \eq{factor}, which can be found in \rf{MTE}.
}

%\Section {pd} Positive Definiteness.

Let us next prove \th{nbpos} establishing the positive definiteness of the $n$ body matrix for nonsingular point configurations.
%Our second main result says that, away from singular configurations, the $n$ body matrix is positive definite and hence defines a Riemannian metric on the space coordinatized by the interpoint distances.
%
%\Th{nbpos} If the $n$ masses $\subs \p n$ do not lie in an affine subspace of dimension $\leq n-2$, then the corresponding $n$ body matrix is positive definite.
%
%\Proofth{nbpos}
Observe that the $n$ body matrix is linear in the inverse mass paramters $\alpha _i$.  Moreover, if we let the mass parameter $m_k = 1$ and send all other $m_j\to \infty$, or, equivalently, $\alpha _k = 1$ and $\alpha _j=0$ for $j \ne k$, then the $n$ body matrix $B = B\ps n$ reduces to the matrix $\dsty B\ps n_k = \Mhat\ps n_k$ obtained by placing the $(i,j)$-th entry of the \kth reduced Euclidean distance matrix $\dsty M\ps n_k\sstrut5$ based at the point $\p_k$ in the position labelled by the unordered index pairs $\set ik$ and $\set jk$, and setting all other entries, \ie those with one or both labels not containing $k$, to zero.  
%Let us call the resulting matrix the \kth \is{expanded Cayley--Menger matrix}. 
We have thus shown that the $n$ body matrix decomposes into a linear combination thereof:
%the expanded Cayley--Menger matrices
\Eq{Bna}
$$B\ps n(\alpha ,r) = \Sum kn \alpha _k \Mhat\ps n_k(r),$$
since each entry of $B\ps n$ is linear in the $\alpha_i$'s.
For example when $n=3$, we write \eq{B3} as
$$\ibeq{1}{B\ps3 =\alpha _1\Mhat\ps3_1+ \alpha _2\Mhat\ps3_2+ \alpha _3\Mhat\ps3_3=\alpha _1\omatrix{2\:r_{12}^2&r_{12}^2 + r_{13}^2 - r_{23}^2&0\sstrut6\\
r_{12}^2 + r_{13}^2 - r_{23}^2&2\: r_{13}^2&0\sstrut6\\
0&0&0} + {}\\
{}\alpha _2\omatrix{2\:r_{12}^2&0&r_{12}^2 + r_{23}^2 - r_{13}^2\sstrut6\\
0&0&0\sstrut6\\
r_{12}^2 + r_{23}^2 - r_{13}^2&0&2\:r_{23}^2}
 + \alpha _3\omatrix{0&0&0\sstrut6\\
0&2\:r_{13}^2&r_{13}^2 + r_{23}^2 - r_{12}^2\sstrut6\\
0&r_{13}^2 + r_{23}^2 - r_{12}^2&2\:r_{23}^2},}$$
and recognize the nonzero entries of its three matrix summands as forming order $3$ reduced Euclidean distance matrices \eq{CM3}.  

Now, to prove positive definiteness of $B = B\ps n$, we need to show positivity of  the associated quadratic form: 
\Eq{posBn}
$$\z^T B \:\z > 0 \roq{for all} \o \ne \z = \ptvector{\ldotsx z_{\set ij} \ldotsx} \in \Rx{n \choose 2}.$$
Using \eq{Bna}, we can similarly expand this quadratic form:
\Eq{posBnx}
$$\z^T B \:\z = \Sum kn \alpha _k \:\z^T\Mhat\ps n_k\:\z = \Sum kn \alpha _k\: \z_k^TM\ps n_k\:\z_k,$$
where
$$\z_k = \ptvector{z_{\set 1k},\ldots,z_{\set {k-1}k},z_{\set {k+1}k},\ldots,z_{\set nk}} \in \Rx{n-1}, \forq \rg kn,$$
so $z_{\set kk}\sstrut5$ is omitted from the vector,  keeping in mind that the indices are symmetric, so $z_{\set ij} = z_{\set ji}$.  The final identity in \eq{posBnx} comes from eliminating all the terms involving the zero entries in $\Mhat\ps n_k$.  Now, \th{CMpos} implies positive semi-definiteness of the reduced Euclidean distance matrices $M\ps n_k$ and, hence,
\Eq{pdMk}
$$\z_k^TM\ps n_k\:\z_k \geq 0,$$
which, by \eq{posBnx}, establishes positive semi-definiteness of the $n$ body matrix.  Moreover, if the point configuration $\subs \p n$ is nonsingular, \th{CMpos} implies positive definiteness of the reduced Euclidean distance matrices, and hence \eq{pdMk} becomes an equality if and only if $\z_k = \o$.  Moreover, if $\z \ne \o \in \Rx{n(n-1)/2}$, then at least one $\z_k \ne \o \in \Rx{n-1}$, and hence at least one of the summands on the right hand side of \eq{posBnx} is strictly positive, which establishes the desired inequality \eq{posBn}, thus proving positive definiteness of the $n$ body matrix. On the other hand, if the configuration is singular, the corresponding Cayley--Menger determinant vanishes, and so the Factorization \th{factor}, to be proved below, implies that the $n$ body determinant also vanishes, which means that the $n$ body matrix cannot be positive definite.\qed

% @Peter: My part begins here.

\Section g Factorization of Certain Determinants.

In order to prove the Factorization \th{factor}, we will, in fact, significantly generalize it.  A proof of  the generalization will establish the desired result as a special case.

%\subsection{\label{subsect.factor1}A tale of three matrices}

\emph{Notation\/}:  For each nonnegative $m\in\mathbb{Z}$, we let $\left[  m\right]  $ be the set $\left\{
1,2,\ldots,m\right\}  $.

Let us  define a class of matrices that includes the Cayley--Menger matrix $C\ps n$ in \eq{lCM} and the matrix $C_A$ in \eq{CA}. 

Let $\Ring$ be a ring. Fix an integer $n\geq1$.  If $H=\left(  h_{i,j}\right)  _{1\leq i\leq n,\ 1\leq j\leq n}\in \Ring^{n\times
n}$ is an $n\times n$-matrix over $\Ring$, then we define an $\left(
n+1\right)  \times\left(  n+1\right)  $-matrix $C_{H}\in \Ring^{\left(
n+1\right)  \times\left(  n+1\right)  }$ by
\Eq{CH}
$$\eeq{C_{H}    =\left(
\begin{cases}
h_{i,j}, & \text{if }i,j\in\left[  n\right]  ;\\
1, & \text{if exactly one of }i,j\text{ belongs to }\left[  n\right]  ;\\
0, & \text{if }i=j=n+1
\end{cases}
\right)  _{1\leq i\leq n+1,\ 1\leq j\leq n+1}\\
 =\left(
\begin{array}
[c]{ccccc}%
h_{1,1} & h_{1,2} & \cdots & h_{1,n} & 1\\
h_{2,1} & h_{2,2} & \cdots & h_{2,n} & 1\\
\vdots & \vdots & \ddots & \vdots & \vdots\\
h_{n,1} & h_{n,2} & \cdots & h_{n,n} & 1\\
1 & 1 & \cdots & 1 & 0
\end{array}
\right)  .}
$$
Observe that our earlier matrices $C\ps n = C_R$, as in \eq{lCM}, and $C_A$, as in \eq{CA}, are both of this form based, respectively, on the $n \times n $ matrices
\Eq{RA}
$$\req{R = \omatrix{0&r_{12}^2&r_{13}^2&\ldots&r_{1n}^2\sstrut6\\
r_{12}^2&0&r_{23}^2&\ldots&r_{2n}^2\sstrut6\\
r_{13}^2&r_{23}^2&0&\ldots&r_{3n}^2\\
\vdots&\vdots&\vdots&\ddots&\vdots\\
%r_{1,n-1}^2&r_{2,n-1}^2&r_{3,n-1}^2&\ldots&r_{n-1,n}^2\\
r_{1n}^2&r_{2n}^2&r_{3n}^2&\ldots&0},\\
A=\omatrix{\alpha_{1} & 0 & 0&\ldots & 0\\
0 & \alpha_{2} & 0& \ldots & 0\\
0 & 0& \alpha_{3} & \ldots & 0\\
\vdots & \vdots &\vdots& \ddots & \vdots \\
0 & 0 & 0&\ldots & \alpha_{n}},}$$
the former being the Euclidean distance matrix \eq{EDM}.

We will work in the polynomial ring
\Eq{Ring}
$$\Ring = \mathbb{Z}\left[  \left\{  s_{i,j}\, \mid\, i,j\in\left[  n\right]  \right\}
\cup\left\{  t_{k,l}\, \mid\, k,l\in\left[  n\right]  \right\}  \right],
$$
consisting of polynomials with integer coefficients depending on the $n^{2}+n^{2} = 2\:n^2$ independent variables $s_{i,j},t_{k,l}$.
Define the corresponding pair of $n\times n$-matrices 
\Eq{ST}
$$\req{S=\left(  s_{i,j}\right)  _{1\leq
i\leq n,\ 1\leq j\leq n}\in
\Ring^{n\times n},\\T=\left(  t_{i,j}\right)  _{1\leq i\leq
n,\ 1\leq j\leq n}\in
\Ring^{n\times n},}$$
which we use to construct the $\left(n+1\right)  \times\left(n+1\right)  $ matrices $C_S$ and $C_T$ via \eq{CH}.  

Next, let $E$ be the set of all $2$-element subsets of $\ive{n}$;
we regard these subsets as unordered pairs of distinct elements of
$\ive{n}$. Note that $\abs{E} = n\left(  n-1\right)/2$.
Our generalization of the $n$ body matrix will be the matrix $W_{S,T}\in \Ring^{E\times E}$ --- that is, a matrix whose
rows and columns are indexed by elements of $E$ --- whose entries are given by
%$W=\left(  w_{\set ij, \set kl }\right)  _{\set ij  \in E,\ \set kl  \in E}$, where
\Eq{wijkl}
$$w_{\set ij  ,\set kl  }=\left(  t_{j,k}+t_{i,l}%
-t_{i,k}-t_{j,l}\right)  \left(  s_{j,k}+s_{i,l}-s_{i,k}-s_{j,l}\right).$$
%for all $\left(  i,j\right)  \in\left[  n\right]  ^{2}$ and $\left(
%k,l\right)  \in\left[  n\right]  ^{2}$ satisfying $i<j$ and $k<l$.
It is easy to see that \eq{wijkl} is well defined for any
$\set ij ,\set kl \in E$, since
% (not necessarily requiring $i<j$ and $k<l$). Indeed, it suffices to check that
 the right hand side is unchanged when $i$ is switched with $j$, and is also
unchanged when $k$ is switched with $l$.
We also remark that the factor $s_{j,k}+s_{i,l}-s_{i,k}-s_{j,l}$ on the right hand
side of \eq{wijkl} can be rewritten as 
$$\det\left(
\begin{matrix}
s_{j,l} & s_{j,k} & 1\\
s_{i,l} & s_{i,k} & 1\\
1 & 1 & 0
\end{matrix}
\right)  =\det\left(  C_{S\left[  j,i\mid l,k\right]  }\right)  ,\roq{where}
S\left[  j,i\mid l,k\right]  = \left(
\begin{matrix}
s_{j,l} & s_{j,k}\\
s_{i,l} & s_{i,k}%
\end{matrix}
\right) ,$$
% where
%$S\left[  j,i\mid l,k\right]  $ is the $2\times2$-matrix $\left(
%\begin{matrix}
%s_{j,l} & s_{j,k}\\
%s_{i,l} & s_{i,k}%
%\end{matrix}
%\right)  $. 
and similarly for the first factor $t_{j,k}+t_{i,l}-t_{i,k}-t_{j,l}$.  Thus, each entry of $W_{S,T}$ is the product of the determinants of a pair of $3 \times 3$ matrices that also have our basic form \eq{CH}.

Since $W_{S,T}$ is a square matrix of size $\abs{E} \times \abs E$, it has a determinant $\det W_{S,T}\in \Ring$.  The main result of this sections is its divisibility:

\Th{Wfactor} We have $\left( \det C_{S}\right) \left( \det  C_{T}\right)
\mid\det W_{S,T}$ in $\Ring$.

Notice that \th{Wfactor} is a divisibility in $\Ring$. Thus, the
quotient is a polynomial $Z_{S,T} = \det W_{S,T}/\left( \det C_{S}\, \det  C_{T}\right)$, with integer coefficients, in the independent variables $s_{i,j},t_{i,j}$.  Thus,
\Eq{dWST}
$$\det W_{S,T} = \left( \det C_{S}\right) \left( \det  C_{T}\right) Z_{S,T}.$$
Observe that, whereas the left hand side of \eq{dWST} depends on all $2\:n^2$ variables, the first factor depends only on the $s_{i,j}$ and the second factor only on the $t_{k,l}$, while the final factor is, in general, a ``mixed'' function of both sets of variables. 
As in \th{factor}, the factorization \eq{dWST} is existential, and we do not have a direct formula for the mixed factor $Z_{S,T}$.  Finding such a formula  and giving it an algebraic or geometric interpretation is an outstanding and very interesting problem.

If we now specialize  $S \mapsto  R$ and $T \mapsto A$, where $R,A$ are the matrices \eq{RA}, then $W_{S,T}$ reduces to 
\Eq{WB}
$$W_{R,A} = - B\ps n,$$
where $B\ps n$ is the $n$ body matrix defined by \eqs{Bd}{Bo}, and thus the left hand side of \eqf{dWST} reduces to the $n$ body determinant 
\Eq{dWB}
$$\det W_{R,A} = (-1)^{n \choose 2}\det B\ps n = (-1)^{n \choose 2}\Delta \ps n.$$  
On the other hand, we use \eq{CMd}  to identify $\det C_R$ with the Cayley--Menger determinant $\delta\ps n$, and \eq{dCA} to identify $\det C_A$ with the negative of the elementary symmetric polynomial $-\,e_{n-1}(\alpha )$.  Thus, the general factorization \eqf{dWST} reduces to the $n$ body determinant factorization \eqf{factor} where the mass-dependent factor
\Eq{sZ} 
$$\sigma \ps n = (-1)^{{n \choose 2} + 1}\,Z_{R,A}$$
 is identified with the corresponding reduction of the mixed factor in \eq{dWST}.  Thus, \th{Wfactor} immediately implies the Factorization \th{factor} upon specialization.  Again, we do not have a direct formula for constructing either factor $Z_{S,T}$ or $Z_{R,A}$.

%\begin{proof}[Proof of \th{factor}.]
%If we substitute
%$s_{i,j}=
%\begin{cases}
%r_{ij}^{2}, & \text{if }i<j;\\
%r_{ji}^{2}, & \text{if }i>j;\\
%0 & \text{if }i=j
%\end{cases}
%$ for all $i,j\in\left[  n\right]  $, then the matrix $C_S$ reduces to
%the matrix $C_R$ in \eq{lCM}, whereas the matrix $\overline{W}$ reduces to
%the $n$ body matrix $B^{\left(  n\right)  }$; therefore, the divisibility
%$\esp{n-1}(\alpha) \det \left(C_S\right) \mid \det \overline{W}$
%from \th{thm.detBa}
%transforms into
%$\esp{n-1}(\alpha) \det\left(  C^{\left(  n\right)  }\right)  \mid\det\left(  B^{\left(
%n\right)  }\right)$,
%that is,
%$\esp{n-1}(\alpha)\, \delta \ps n(r) \mid
%\Delta^{\left(  n\right)  }$, which proves \th{factor}.
%\end{proof}

Our proof of \th{Wfactor} will rely on basic properties of
UFDs (unique factorization domains), which are found in most texts on
abstract algebra, e.g., \cite[Section VIII.4]{Knapp}.
We shall also use the fact that any polynomial ring (in finitely many
variables) over $\mathbb{Z}$ is a UFD. (This follows, \textit{e.g.},
from \cite[Corollary 8.21]{Knapp} by induction on the number of
variables.)
% This is also true for infinitely many variables, but it's harder to
% find a reference for it.
Moreover, we shall use the fact (obvious from degree considerations)
that the only units (\textit{i.e.}, invertible elements)
of a polynomial ring are constant polynomials.
Hence, a polynomial $p$ in a polynomial ring
$\mathbb{Z}\left[x_1, x_2, \ldots, x_k\right]$
is irreducible if its \emph{content}, \textit{i.e.}, the gcd of its coefficients, is $1$ and $p$ is irreducible in
the ring $\mathbb{Q}\left[x_1, x_2, \ldots, x_k\right]$
(since any constant factor of $p$ in
$\mathbb{Z}\left[x_1, x_2, \ldots, x_k\right]$
must divide the content of $p$).

Before we prove \th{Wfactor}, we require a technical lemma:

\begin{Lemma}
\label{lem.detC-irr}Assume that $n>1$. Then, $\det C_S  $ is a
prime element of the UFD $\Ring$.
\end{Lemma}

\begin{proof}[Proof of Lemma \ref{lem.detC-irr}.]
Expanding $\det C_S$ as a sum over all
$\left(  n+1\right)  !$ permutations $\pi$ of $\left[
n+1\right]  $, we observe that the permutations $\pi$ satisfying
$\pi\left(  n+1\right)  =n+1$ give rise to summands that equal $0$, whereas
all the other permutations $\pi$ contribute \is{pairwise distinct
monomials} to the sum\footnote{Why pairwise distinct?
The monomial corresponding to such a permutation $\pi$ is
$\prod_{i \in \left[n\right];\  \pi\left(i\right) \neq n+1} s_{i, \pi\left(i\right)}$.
Knowing this monomial, we can reconstruct the value of $\pi$
at the unique $i$ satisfying $\pi\left(i\right) = n + 1$
(namely, this value is the unique $k \in \left[n\right]$
for which no entry from the $k$-th row of $S$ appears in
the monomial), as well as the remaining values of $\pi$
on $\left[n\right]$ (by inspecting the corresponding $s_{i, j}$
in the monomial), and finally the value of $\pi$ at
$n+1$ (as the remaining element of $\left[n+1\right]$).
Thus, we can reconstruct $\pi$ uniquely from this monomial.}.
This shows that the polynomial $\det C_S  $ has
content $1$; indeed, each of its nonzero coefficients is $1$ or $-1$.
Moreover, it shows that
$\det C_S$ is a polynomial of degree $1$ in each
of the indeterminates $s_{i, j}$
(not $0$ because $n > 1$).
Furthermore,
in the expansion of $\det C_S$ into monomials, each monomial
contains at most one variable from each row and at most one from each column.
Thus, the same argument that is used in \cite[proof of Lemma 5.12]{DeKuRo17}
to prove the irreducibility of $\det S$ can be used to see that
$\det C_S$ is an irreducible element
of the ring $\mathbb{Q} \left[ s_{i,j} \mid i, j \in \left[ n \right] \right]$.
Hence, since $\det C_S$ has content $1$, it is an irreducible element of the ring
$\Ring _0 = \mathbb{Z} \left[ s_{i,j} \mid i, j \in \left[ n \right] \right]$
as well.
Hence, $\det C_S$ is also an irreducible element of the ring $\Ring$ (which
differs from $\Ring _0 $ merely in the introduction of $n^{2}$ new variables $t_{i,j}$, which clearly do not contribute any possible divisors to $\det C_S$). Since $\Ring$ is a UFD, we thus conclude that $\det C_S$ is a prime element of $\Ring$.
\end{proof}

\begin{proof}
[Proof of \th{Wfactor}.]If $n=1$, then \th{Wfactor} is
clear, since $\det C_S  =-1$ and $\det C_{T}=-1$
in this case. Thus, without loss of generality, assume that $n>1$.

Since $\Ring$ is a polynomial ring over $\mathbb{Z}$, it is a UFD.
Moreover, Lemma \ref{lem.detC-irr} yields that $\det C_S  $ is
a prime element of $\Ring$. Similarly, $\det C_{T}$ is a prime
element of $\Ring$.

Let $\Quo = \Ring/\det C_S $ be the quotient ring, which is an integral domain since $\det C_S  $ is a prime
element of $\Ring$.
Since, by construction, $\det C_S  =0$ in $\Quo$, the matrix $C_{S}$ is singular over $\Quo$ and hence has a nontrivial kernel because $\Quo$ is an
integral domain. In other words, there exists a nonzero vector $\o \ne \x^* = \left(
x_{1},x_{2},\ldots,x_{n},v\right)  ^{T}\in \Quo^{n+1}$ such that
\Eq{CQx0} 
$$C_{S}\x^*=\o.$$
The entries of the vector identity \eq{CQx0} imply\fnote{Here and in the following, \textquotedblleft$\sum_{l}$\textquotedblright%
\ always means \textquotedblleft$\sum_{l=1}^{n}$\textquotedblright.}
\Eq{pft4}
$$\qeq{\sum_{l}s_{i,l}x_{l}+v=0\\\text{for all }i\in\left[  n\right],\\\quad
\text{and} \qquad
\sum_{l}x_{l}=0.}$$
%We can
%rewrite \eqref{pf.t4.4a} as follows:
%\begin{equation}
%\sum_{l}s_{i,l}x_{l}=-y\ \ \ \ \ \ \ \ \ \ \text{for all }i\in\left\{
%1,2,\ldots,n\right\}  .
%\end{equation}
Given such an $\x^*$, let $\mathbf{x} = \left(  x_{1},x_{2},\ldots,x_{n}\right)^{T}\in \Quo^{n}$ be the vector obtained by omitting the last entry. If $\x=\o$, then, according to the first equations in \eq{pft4}, this would require $v = 0$, which would contradict the fact that $\x^* \ne \o$.  Thus, $\x \ne \o$, which, by the last equation in \eq{pft4}, implies
that $\x$ has at least two nonzero entries, so $x_{i}x_{j} \ne 0$ for some $i \ne j$, since $\Quo$ is an integral domain.

Define the vector $\z\in \Quo^E$ whose entries are  indexed by unordered pairs $\set ij  \in E$ and given by the products of distinct entries of $\x$, so 
%$z_{\set ij} = x_{i}x_{j}$ for $\set ij  \in E$.
$$\req{z_{\set ij} = x_{i}x_{j}, \\ \set ij  \in E.}$$
Hence, $\z\neq\o$ (since $x_{i}x_{j} \ne 0$ for some $i \ne j$).

Let us abbreviate $W = W_{S,T}$. We shall prove that $\o \ne \z \in \ker W$. To this end, for any $1\leq i<j\leq n$, the $\set ij  $-th entry of the vector $W\z$ is
\begin{align*}
&  \underbrace{\sum_{\set kl  \in E}}_{=\sum_{k<l}}%
\underbrace{w_{\set ij  ,\set kl  }}_{\substack{=\left(
t_{j,k}+t_{i,l}-t_{i,k}-t_{j,l}\right)  \left(  s_{j,k}+s_{i,l}-s_{i,k}%
-s_{j,l}\right)  \\\text{(by \eq{wijkl})}}}x_{k}x_{l}\\
&  =\sum_{k<l}\underbrace{\left(  t_{j,k}+t_{i,l}-t_{i,k}-t_{j,l}\right)
}_{=\left(  t_{j,k}-t_{i,k}\right)  -\left(  t_{j,l}-t_{i,l}\right)  }\left(
s_{j,k}+s_{i,l}-s_{i,k}-s_{j,l}\right)  x_{k}x_{l}\\
&  =\sum_{k<l}\left(  \left(  t_{j,k}-t_{i,k}\right)  -\left(  t_{j,l}%
-t_{i,l}\right)  \right)  \left(  s_{j,k}+s_{i,l}-s_{i,k}-s_{j,l}\right)
x_{k}x_{l}\\
&  =\sum_{k<l}\left(  t_{j,k}-t_{i,k}\right)  \left(  s_{j,k}+s_{i,l}%
-s_{i,k}-s_{j,l}\right)  x_{k}x_{l}-\sum_{k<l}\left(  t_{j,l}-t_{i,l}\right)
\underbrace{\left(  s_{j,k}+s_{i,l}-s_{i,k}-s_{j,l}\right)  }_{=-\left(
s_{j,l}+s_{i,k}-s_{i,l}-s_{j,k}\right)  }
\underbrace{x_{k}x_{l}}_{= x_l x_k}\\
&  =\sum_{k<l}\left(  t_{j,k}-t_{i,k}\right)  \left(  s_{j,k}+s_{i,l}%
-s_{i,k}-s_{j,l}\right)  x_{k}x_{l}+\sum_{k<l}\left(  t_{j,l}-t_{i,l}\right)
\left(  s_{j,l}+s_{i,k}-s_{i,l}-s_{j,k}\right)  x_l x_k\\
&  =\sum_{k<l}\left(  t_{j,k}-t_{i,k}\right)  \left(  s_{j,k}+s_{i,l}%
-s_{i,k}-s_{j,l}\right)  x_{k}x_{l}+\sum_{k>l}\left(  t_{j,k}-t_{i,k}\right)
\left(  s_{j,k}+s_{i,l}-s_{i,k}-s_{j,l}\right)  x_{k}x_{l}\\
&  \qquad\left(  \text{here, we switched the roles of }k\text{ and }l\text{ in
the second sum}\right) \\
&  =\sum_{k,l\in\left[  n\right]  }\left(  t_{j,k}-t_{i,k}\right)  \left(
s_{j,k}+s_{i,l}-s_{i,k}-s_{j,l}\right)  x_{k}x_{l}\\
&  \qquad\left(
\begin{array}
[c]{c}%
\text{here, we have combined the two sums, while also including extra}\\
\text{terms for }k=l\text{ (which don't change the sum since they are
}0\text{)}%
\end{array}
\right) \\
&  =\sum_{k\in\left[  n\right]  }\left(  t_{j,k}-t_{i,k}\right)
x_{k}{\sum_{l\in\left[  n\right]  }\left(  s_{j,k}+s_{i,l}%
-s_{i,k}-s_{j,l}\right)  x_{l}}
%_{=s_{j,k}\sum_{l\in\left[  n\right]  }%
%x_{l}+\sum_{l\in\left[  n\right]  }s_{i,l}x_{l}-s_{i,k}\sum_{l\in\left[
%n\right]  }x_{l}-\sum_{l\in\left[  n\right]  }s_{j,l}x_{l}}
\\
&  =\sum_{k\in\left[  n\right]  }\left(  t_{j,k}-t_{i,k}\right)  x_{k}\left(
s_{j,k}\underbrace{\sum_{l\in\left[  n\right]  }x_{l}}%
_{\substack{=0}}+\underbrace{\sum_{l\in\left[
n\right]  }s_{i,l}x_{l}}_{\substack{=-v}%
}-s_{i,k}\underbrace{\sum_{l\in\left[  n\right]  }x_{l}}%
_{\substack{=0}}-\underbrace{\sum_{l\in\left[
n\right]  }s_{j,l}x_{l}}_{\substack{=-v%
}}\right) \\
&  \qquad\left(
\begin{array}
[c]{c}%
\text{here, we have used the equations in \eq{pft4} on each set of terms}%
\end{array}
\right) 
\end{align*}%
\vglue-20pt
\begin{align*}
&  =\sum_{k\in\left[  n\right]  }\left(  t_{j,k}-t_{i,k}\right)
x_{k}\underbrace{\left(  s_{j,k}0+\left(  -v\right)  -s_{i,k}0-\left(
-v\right)  \right)  }_{=0}=0.\hskip150pt \,
\end{align*}
Hence, $W\z=\o$. Since $\z$ is a nonzero vector, this
shows that $W$ has a nontrivial kernel over $\Quo$. Since $\Quo$ is an
integral domain, we thus conclude that $\det W=0$ in $\Quo$. In other words,
$\det C_S  \mid\det W$. The same argument shows that
$\det C_T \mid\det W$ also, since $S$ and $T$ play symmetric roles
in the definition of the matrix $W$.

Finally, we note that the two prime elements $\det C_S  $ and $\det C_T $ of $\Ring$ are distinct --- indeed, they are polynomials in disjoint
sets of indeterminates $s_{i,j}$ and $t_{k,l}$, respectively, so
they could only be equal if they were both constant, which they are not.
Thus, they are coprime. Hence, an element of the UFD $\Ring$ divisible both by
$\det C_S  $ and by $\det C_T $ must also be
divisible by their product $\det C_S  \det C_T  $. Applying this to the element $\det W$ completes the proof of the General Factorization \th{Wfactor}.
\end{proof}

A further generalization of \th{Wfactor} will be given in the upcoming work \rf{Grinberg}.

%\Section q Generalizing \th{Q}.

\Section q A Biquadratic Form Identity.

In this section we establish a striking identity involving the matrix $W_{S,T}$, which naturally defines a biquadratic form that factorizes over a particular pair of hyperplanes.  The reduction of this formula to the $n$ body matrix is also of note.

As above, let $W = W_{S,T} \in \Ring^{E\times E}$ be the $\abs{E} \times \abs E$ matrix  whose entries are given by \eq{wijkl}.  Let $\Alg$ be a commutative $\Ring$-algebra. Define the biquadratic form 
\Eq{q22}
$$Q_W(\x,\y) = \sum_{\set ij  \in E}\sum_{\set kl  \in E}
w_{\set ij ,\set kl  } x_{i}x_{j}y_{k}y_{l},$$
where $\x=\left(  x_{1},x_{2},\ldots,x_{n}\right)  ^{T}$
and $\y
=\left(  y_{1},y_{2},\ldots,y_{n}\right)  ^{T}$ are vectors in $\Alg  ^{n}$.

\Th{quartic-gen} When $x_{1}+x_{2}+\cdots+x_{n}=0$ and
$y_{1}+y_{2}+\cdots+y_{n}=0$,  the biquadratic form
\eq{q22} factors into a product of two elementary bilinear forms\fnote{The ${}^T$ superscripts are transposes of column vectors, and have nothing to do with the matrix $T$.} based on the matrices $S,T$ given in \eq{ST}:
\Eq{q22f} 
$$Q_W(\x,\y) =  (  \x^{T}\,S\,\y)  \,( \x^{T}\,T\,\y)  .
$$

%Applying this to $\mathbf{y}=\mathbf{x}$ and using the same specialization
%$W\rightarrow\overline{W}\rightarrow B^{\left(  n\right)  }$ as in Subsection
%\ref{subsect.factor2}, we recover \th{Q}.

\begin{proof}
We calculate
\begin{align*}
& Q_W(\x,\y) =   \underbrace{\sum_{\set ij  \in E}}_{=\sum_{i<j}}\underbrace{\sum
_{\set kl  \in E}}_{=\sum_{k<l}}\underbrace{w_{\set ij
,\set kl  }}_{\substack{=\left(  t_{j,k}+t_{i,l}-t_{i,k}%
-t_{j,l}\right)  \left(  s_{j,k}+s_{i,l}-s_{i,k}-s_{j,l}\right)  \\\text{(by
\eq{wijkl})}}}x_{i}x_{j}y_{k}y_{l}\\
&  =\sum_{i<j}\sum_{k<l}\left(  t_{j,k}+t_{i,l}-t_{i,k}-t_{j,l}\right)
\left(  s_{j,k}+s_{i,l}-s_{i,k}-s_{j,l}\right)  x_{i}x_{j}y_{k}y_{l}\\
&  =\sum_{i<j}\sum_{k<l}t_{j,k}\left(  s_{j,k}+s_{i,l}-s_{i,k}-s_{j,l}\right)
x_{i}x_{j}y_{k}y_{l}+\underbrace{\sum_{i<j}\sum_{k<l}t_{i,l}\left(
s_{j,k}+s_{i,l}-s_{i,k}-s_{j,l}\right)  x_{i}x_{j}y_{k}y_{l}}_{\substack{=\sum
_{i>j}\sum_{k>l}t_{j,k}\left(  s_{i,l}+s_{j,k}-s_{j,l}-s_{i,k}\right)
x_{j}x_{i}y_{l}y_{k}\\\text{(here, we have swapped }i\text{ with }j\text{,}\\
\text{and also swapped }k\text{ with }l\text{)}}}\\
&  \qquad-\underbrace{\sum_{i<j}\sum_{k<l}t_{i,k}\left(  s_{j,k}%
+s_{i,l}-s_{i,k}-s_{j,l}\right)  x_{i}x_{j}y_{k}y_{l}}_{\substack{=\sum
_{i>j}\sum_{k<l}t_{j,k}\left(  s_{i,k}+s_{j,l}-s_{j,k}-s_{i,l}\right)
x_{j}x_{i}y_{k}y_{l}\\\text{(here, we have swapped }i\text{ with }j\text{)}%
}}-\underbrace{\sum_{i<j}\sum_{k<l}t_{j,l}\left(  s_{j,k}+s_{i,l}%
-s_{i,k}-s_{j,l}\right)  x_{i}x_{j}y_{k}y_{l}}_{\substack{=\sum_{i<j}%
\sum_{k>l}t_{j,k}\left(  s_{j,l}+s_{i,k}-s_{i,l}-s_{j,k}\right)  x_{i}%
x_{j}y_{l}y_{k}\\\text{(here, we have swapped }k\text{ with }l\text{)}}}\\
&  =\sum_{i<j}\sum_{k<l}t_{j,k}\left(  s_{j,k}+s_{i,l}-s_{i,k}-s_{j,l}\right)
x_{i}x_{j}y_{k}y_{l}+\sum_{i>j}\sum_{k>l}t_{j,k}\underbrace{\left(
s_{i,l}+s_{j,k}-s_{j,l}-s_{i,k}\right)  }_{=s_{j,k}+s_{i,l}-s_{i,k}-s_{j,l}%
}\underbrace{x_{j}x_{i}y_{l}y_{k}}_{=x_{i}x_{j}y_{k}y_{l}}\\
&  \quad-\sum_{i>j}\sum_{k<l}t_{j,k}\underbrace{\left(  s_{i,k}%
+s_{j,l}-s_{j,k}-s_{i,l}\right)  }_{=-\left(  s_{j,k}+s_{i,l}-s_{i,k}%
-s_{j,l}\right)  }\underbrace{x_{j}x_{i}}_{=x_{i}x_{j}}y_{k}y_{l}-\sum
_{i<j}\sum_{k>l}t_{j,k}\underbrace{\left(  s_{j,l}+s_{i,k}-s_{i,l}%
-s_{j,k}\right)  }_{=-\left(  s_{j,k}+s_{i,l}-s_{i,k}-s_{j,l}\right)  }%
x_{i}x_{j}\underbrace{y_{l}y_{k}}_{=y_{k}y_{l}}\\
&  =\sum_{i<j}\sum_{k<l}t_{j,k}\left(  s_{j,k}+s_{i,l}-s_{i,k}-s_{j,l}\right)
x_{i}x_{j}y_{k}y_{l}+\sum_{i>j}\sum_{k>l}t_{j,k}\left(  s_{j,k}+s_{i,l}%
-s_{i,k}-s_{j,l}\right)  x_{i}x_{j}y_{k}y_{l}\\
&  \quad+\sum_{i>j}\sum_{k<l}t_{j,k}\left(  s_{j,k}+s_{i,l}-s_{i,k}%
-s_{j,l}\right)  x_{i}x_{j}y_{k}y_{l}+\sum_{i<j}\sum_{k>l}t_{j,k}\left(
s_{j,k}+s_{i,l}-s_{i,k}-s_{j,l}\right)  x_{i}x_{j}y_{k}y_{l}\\
&  =\sum_{i\neq j}\sum_{k\neq l}t_{j,k}\left(  s_{j,k}+s_{i,l}-s_{i,k}%
-s_{j,l}\right)  x_{i}x_{j}y_{k}y_{l}\\
&  \qquad\left(  \text{here, we have combined all four sums into a single
one}\right)  \\
&  =\sum_{i,j}\sum_{k,l}t_{j,k}\left(  s_{j,k}+s_{i,l}-s_{i,k}-s_{j,l}\right)
x_{i}x_{j}y_{k}y_{l}\\
&  \qquad\left(
\begin{array}
[c]{c}%
\text{here, we have inserted extraneous addends for }i=j\text{ and for
}k=l\text{,}\\
\text{which are }0\text{ and therefore don't change our sum}%
\end{array}
\right)\\
&  =\underbrace{\sum_{i,j}\sum_{k,l}t_{j,k}s_{j,k}x_{i}x_{j}y_{k}y_{l}%
}_{=\left(  \sum_{i}x_{i}\right)  \left(  \sum_{j,k,l}t_{j,k}s_{j,k}x_{j}%
y_{k}y_{l}\right)  }+\underbrace{\sum_{i,j}\sum_{k,l}t_{j,k}s_{i,l}x_{i}%
x_{j}y_{k}y_{l}}_{=\left(  \sum_{i,l}x_{i}s_{i,l}y_{l}\right)  \left(
\sum_{j,k}x_{j}t_{j,k}y_{k}\right)  }\\
&  \qquad-\underbrace{\sum_{i,j}\sum_{k,l}t_{j,k}s_{i,k}x_{i}x_{j}y_{k}y_{l}%
}_{=\left(  \sum_{l}y_{l}\right)  \left(  \sum_{i,j,k}t_{j,k}s_{i,k}x_{i}%
x_{j}y_{k}\right)  }-\underbrace{\sum_{i,j}\sum_{k,l}t_{j,k}s_{j,l}x_{i}%
x_{j}y_{k}y_{l}}_{=\left(  \sum_{i}x_{i}\right)  \left(  \sum_{j,k,l}%
t_{j,k}s_{j,l}x_{j}y_{k}y_{l}\right)  }\\
\end{align*}%
\begin{align*}
&  =\underbrace{\left(  \sum_{i}x_{i}\right)  }_{=x_{1}+x_{2}+\cdots+x_{n}%
=0}\left(  \sum_{j,k,l}t_{j,k}s_{j,k}x_{j}y_{k}y_{l}\right)  +\left(
\sum_{i,l}x_{i}s_{i,l}y_{l}\right)  \left(  \sum_{j,k}x_{j}t_{j,k}%
y_{k}\right)  \\
&  \qquad-\underbrace{\left(  \sum_{l}y_{l}\right)  }_{=y_{1}+y_{2}%
+\cdots+y_{n}=0}\left(  \sum_{i,j,k}t_{j,k}s_{i,k}x_{i}x_{j}y_{k}\right)
-\underbrace{\left(  \sum_{i}x_{i}\right)  }_{=x_{1}+x_{2}+\cdots+x_{n}%
=0}\left(  \sum_{j,k,l}t_{j,k}s_{j,l}x_{j}y_{k}y_{l}\right)  \\
&  =\underbrace{\left(  \sum_{i,l}x_{i}s_{i,l}y_{l}\right)  }_{=\mathbf{x}%
^{T}S\mathbf{y}}\ \underbrace{\left(  \sum_{j,k}x_{j}t_{j,k}y_{k}\right)
}_{=\mathbf{x}^{T}T\mathbf{y}}=\left(  \mathbf{x}^{T}S\mathbf{y}\right)
\cdot\left(  \mathbf{x}^{T}T\mathbf{y}\right)  .
\end{align*}
\vglue-20pt
\end{proof}

Let us now specialize the identity in \th{quartic-gen} to the $n$ body case, so that $C_S, C_T$ reduce, respectively, to $C_A,C_R$. We further set $\x = \y$ and then use \eq{qks}, \eq{qedm}, \eq{qqk}, \eq{MA}, \eq{AP}, and \eq{WB} to obtain the following intriguing result, reminiscent of Schoenberg's formula \eq{qqk} for the quadratic forms based on the (reduced) Euclidean distance matrices.

%Starting with the $n$ body matrix $B \ps n $ (with variable masses $m_i$), form the following homogeneous quartic polynomial:

\Th Q Define the  quadratic forms
\Eq{rpqf}
$$\req{a(\x) = \Sum in \alpha _i x_i^2 = \Sum in \frac{x_i^2}{m_i},\\
p(\x) = \Sumu{i,j} (\p_i\cdot \p_j) \:x_ix_j = \nnorm{P\,\x }^2,}$$
where $P = \rvectors \p n$ is the point configuration matrix, as in \eq{AP}.
% whose columns are the locations $\p_i$ of the masses. 
%\cf \eq P.
% @Peter: I've tried rewriting this so that it doesn't rely on the deleted equation \eq{P} anymore.
Given the $n$ body matrix $B = B\ps n(\alpha ,r)$ with entries \eqs{Bd}{Bo}, define the corresponding homogeneous quartic form\fnote{The sum is over all ordered pairs $(\set ij, \set kl)$ of unordered pairs $\set ij$ and $\set kl$.} 
\Eq{Q4}
$$\req{Q_B(\x) = \Sumu{\set ij,\set kl} b_{\set ij,\set kl}\: x_ix_jx_kx_l,\\\x = (x_1,\ldotsx, x_n)^T.}$$
Then, when restricted to the hyperplane $\H$, the $n$ body quartic $Q_B(\x)$ factors as the product of the preceding quadratic forms \eq{rpqf}\/\ro:
\Eq{Qf}
$$Q_B(\x) = 2\:a(\x)\,p(\x)\soq{when} x_1 + x_2 + \cdotsx + x_n  = 0.$$

Note that because $b_{\set ij,\set kl} = 0\sstrut5$ when $i,j,k,l$ are distinct, the $n$ body matrix $B\ps n$ is, in fact uniquely determined by its associated quartic form $Q_B(\x)$.  
%Note also the similarity between \eq{Qf} and .
%  A straightforward  calculation establishes the following remarkably simple factored expression for the $n$ body quartic polynomial.

If the point configuration is nonsingular, then, in view of \eq{AP}, the right hand side of \eq{Qf} is positive whenever $\o \ne \x \in \H$, which implies $Q_B(\x) > 0$ under the same conditions.  However, this result does not lead to the conclusion that the $n$ body matrix, which forms the coefficients of $Q_B(\x)$, is itself positive definite, and hence we needed a different approach to establish this result.

\Section c Future Directions.

As noted above, the challenge now is to determine an explicit geometrical formula for the mass-dependent factor $\sigma\ps n(\alpha ,r)$ in the $n$ body determinant factorization \eqf{factor} or, more generally, the mixed factor $Z_{S,T}$ in our generalized factorization \eqf{dWST}, and to ascertain its algebraic and/or geometric significance. Is there some as yet undetected interesting determinantal identity or algebraic structure, perhaps representation-theoretic, that underlies this factorization and its compound determinantal generalization found in \rf{Grinberg}?  Do the biquadratic and quartic form identities we found in \eqs{q22f}{Qf} provide any additional insight into these issues?

 Another important open problem is to understand the geometric structure of the associated Riemannian manifold $\Cno\sstrut8$  that prescribes the radial $n$ body Laplace--Beltrami operator constructed in \rf{MTE} on the Euclidean distance cone, and its consequences for the classical and quantum $n$ body problems.

\vskip30pt

\Ack The second author thanks Alexander Turbiner, Willard Miller, Jr., and Adrian Escobar-Ruiz for introducing him to this problem and for helpful discussions and much needed encouragement during his initial attempts to prove the conjecture.
Both authors thank Victor Reiner for enlightening discussions.  We also thank the referees for their suggestions, corrections, and inspiration to further investigate the intriguing history of Cayley--Menger determinants and matrices.

\vfill\eject

%%%%%%%%%%%%%%%%%%%%%%%%%%%%%

\end{document}